\documentclass[conference]{IEEEtran}

\usepackage{amsmath,amssymb,amsthm}
\usepackage{mathtools}
\usepackage{bm}
\usepackage{graphicx}
\usepackage{cite}

\usepackage[top=0.73in,bottom=0.80in,left=0.73in,right=0.73in]{geometry}

\newtheorem{theorem}{Theorem}
\newtheorem{lemma}{Lemma}

\newtheorem{corollary}{Corollary}
\theoremstyle{definition}
\newtheorem{definition}{Definition}
\theoremstyle{remark}

\title{Explicit Entropic Constructions for Coverage, Facility Location, and Graph Cuts}

\author{
\IEEEauthorblockN{Rishabh Iyer}
\IEEEauthorblockA{
Computer Science Department\\
The University of Texas at Dallas\\
Richardson, TX, USA\\
rishabh.iyer@utdallas.edu
}
}

\begin{document}
\maketitle

\begin{abstract}
Shannon entropy is a polymatroidal set function and lies at the foundation of information theory, yet the class of entropic polymatroids is strictly smaller than the class of all submodular functions.
In parallel, \emph{submodular and combinatorial information measures} (SIMs) have recently been proposed as a principled framework for extending entropy, mutual information, and conditional mutual information to general submodular functions, and have been used extensively in data subset selection, active learning, domain adaptation, and representation learning.
This raises a natural and fundamental question: \emph{are the monotone submodular functions most commonly used in practice entropic}?

In this paper, we answer this question in the affirmative for a broad class of widely used polymatroid functions.
We provide explicit entropic constructions for set cover and coverage functions, facility location, saturated coverage, concave-over-modular functions via truncations, and monotone graph-cut–type objectives.
Our results show that these functions can be realized exactly as Shannon entropies of appropriately constructed random variables.
As a consequence, for these functions, submodular mutual information coincides with classical mutual information, conditional gain specializes to conditional entropy, and submodular conditional mutual information reduces to standard conditional mutual information in the entropic sense.
These results establish a direct bridge between combinatorial information measures and classical information theory for many of the most common submodular objectives used in applications.
\end{abstract}

\section{Introduction}

\subsection{Background: submodularity and entropy}
Let $V$ be a finite ground set. A set function $f:2^V\to \mathbb{R}$ is \emph{submodular} if it satisfies the
diminishing returns property: for all $A\subseteq B\subseteq V$ and $i\in V\setminus B$~\cite{Fujishige2005,bilmes2022submodularity},
\begin{equation}
f(A\cup\{i\})-f(A)\ \ge\ f(B\cup\{i\})-f(B).
\label{eq:dr_submod}
\end{equation}
Equivalently, $f$ is submodular iff for all $A,B\subseteq V$,
\begin{equation}
f(A)+f(B)\ \ge\ f(A\cup B)+f(A\cap B).
\label{eq:submod_lattice}
\end{equation}
Submodular functions are central in combinatorial optimization and have a rich polyhedral structure
~\cite{Edmonds1970,Fujishige2005,Schrijver2003,bilmes2022submodularity,iyer2015polyhedral}.

Shannon entropy provides the canonical example of a submodular function.
For jointly distributed random variables $X_1,\ldots,X_n$, define $h(A):=H(X_A)$ where
$X_A := (X_i)_{i\in A}$. Then $h$ is normalized, monotone, and submodular ~\cite{Shannon1948,Yeung2008}.
This entropic--submodular connection is foundational in information theory; however, not every submodular function
is entropic, and characterizing the gap between the entropy region and the polymatroid cone is a central topic
~\cite{ZhangYeung1997,ZhangYeung1998,DoughertyFreilingZeger2011,Matus2013AlmostEntropic}.

\subsection{Information measures and their submodular counterparts}
For random variables, classical information measures can be written as linear combinations of entropies:
mutual information $I(X;Y)=H(X)+H(Y)-H(X,Y)$ and conditional entropy
$H(X\mid Y)=H(X,Y)-H(Y)$ \cite{Shannon1948,Yeung2008}.
Motivated by this algebraic structure, \emph{submodular and combinatorial information measures} (SIMs)
extend these primitives to general submodular functions by substituting $H(\cdot)$ with an arbitrary submodular
$f(\cdot)$ \cite{iyer2021generalized,iyer2021submodular,asnani2021independence}.
For example, given a set function $f:2^V\to \mathbb{R}$ and disjoint sets $A,B,C\subseteq V$, the
\emph{submodular mutual information}, \emph{conditional gain} (submodular conditional entropy), and
\emph{submodular conditional mutual information} are defined respectively as
\begin{align}
I_f(A;B) &:= f(A)+f(B)-f(A\cup B), \label{eq:smi_def}\\
H_f(A\mid B) &:= f(A\cup B)-f(B), \label{eq:cond_gain_def}\\
I_f(A;B\mid C) &:= f(A\cup C)+f(B\cup C) \nonumber \\
                &-f(C)-f(A\cup B\cup C). \label{eq:scmi_def}
\end{align}
These measures retain many information-like properties (e.g., nonnegativity under mild conditions, chain rules,
and useful optimization structure) while applying to a much larger class of objectives
\cite{iyer2021generalized,asnani2021independence}.

SIMs have been used extensively in data-centric learning, including active learning in realistic regimes
\cite{kothawade2021similar, kothawade2022talisman}, guided data subset selection \cite{kothawade2022prism}, selection under distribution shift
\cite{karanam2022orient}, semi-supervised and meta-learning \cite{li2022platinum}, and representation learning via
combinatorial objectives \cite{majee2024score, majee2024smile}. Recent theory further analyzes SIMs in targeted subset selection settings
\cite{beck2024theoretical}.

\subsection{Question and contributions}
Across these applications, many effective SIM instantiations are built from a small family of classical submodular
functions used as modeling primitives, notably coverage/set cover, facility location, truncations (concave-over-modular),
and graph-based objectives \cite{Edmonds1970,Fujishige2005,Schrijver2003}.
This motivates a fundamental question:

\begin{quote}
\emph{Are the submodular functions most commonly used in applications entropic?}
\end{quote}

Equivalently, given a practical submodular objective $f:2^V\to\mathbb{R}$, does there exist a collection of random variables
$\{X_i\}_{i\in V}$ such that $f(A)=H(X_A)$ for all $A\subseteq V$ (up to a fixed scaling)?
A priori, this need not hold since the entropic region is strictly smaller than the polymatroid cone
\cite{ZhangYeung1998,DoughertyFreilingZeger2011}.

\textbf{Main result.}
We answer this question in the affirmative for a broad class of widely used objectives by giving explicit, constructive
entropic realizations for:
(i) set cover / weighted coverage,
(ii) facility location,
(iii) saturated coverage,
(iv) concave-over-modular functions via truncations (including cardinality and weighted truncations), and
(v) a monotone graph-cut--type objective.
Our constructions are simple and modular, relying on interpretable ``shared'' and ``private'' random components.

\textbf{Implication for SIM.}
For these function classes, the SIM primitives in \eqref{eq:smi_def}--\eqref{eq:scmi_def} become \emph{exact instances}
of classical information measures: submodular mutual information becomes mutual information, conditional gain becomes
conditional entropy, and submodular conditional mutual information becomes conditional mutual information in the entropic
sense. Thus, for a large portion of submodular objectives used in recent application pipelines
\cite{kothawade2021similar,kothawade2022prism,karanam2022orient,li2022platinum,majee2024score,majee2024smile}, the combinatorial framework
admits a direct Shannon-theoretic interpretation.

\textbf{Organization.}
Section~\ref{sec:entropy_polyhedral} summarizes relevant background on the entropy region and polyhedral polymatroids.
Section~\ref{sec:entropic_constructions} then presents explicit entropic constructions for each function family. In the rest of the paper, we provide entropic construction for several commonly used submodular functions: set cover (Sec~\ref{sec:set_cover}), facility location (Sec~\ref{sec:fac_loc}), concave over modular functions (Sec~\ref{sec:concavemod}), saturated coverage (Sec~\ref{sec:saturatedcov}), and graph cut (Sec~\ref{sec:graph-cut}).


\section{Entropy Region and Polyhedral Polymatroids}
\label{sec:entropy_polyhedral}

Entropy functions form a structured subset of polymatroids.
Let $\Gamma_n$ denote the \emph{polymatroid cone} on $n$ elements, i.e., the set of all set functions
$h:2^{[n]}\to\mathbb{R}$ that are normalized ($h(\emptyset)=0$), monotone, and submodular.
It is well known that $\Gamma_n$ is a polyhedral cone described by the elemental (Shannon-type) inequalities
(equivalently, submodularity inequalities) \cite{Yeung2008,Fujishige2005,Schrijver2003,iyer2015polyhedral}.

Let $\Gamma_n^\ast$ denote the set of \emph{entropic} functions, i.e., those $h$ for which there exist random variables
$X_1,\ldots,X_n$ such that $h(A)=H(X_A)$ for all $A\subseteq[n]$ \cite{Shannon1948,Yeung2008}.
The closure $\overline{\Gamma_n^\ast}$ is the \emph{almost-entropic} region.
A central theme in information theory is to understand $\Gamma_n^\ast$ and its relationship to $\Gamma_n$.

For $n\le 3$, entropic functions are completely characterized by Shannon-type inequalities, and hence
$\overline{\Gamma_n^\ast}=\Gamma_n$. In contrast, for $n\ge 4$, the inclusion is strict:
\[
\overline{\Gamma_n^\ast}\ \subsetneq\ \Gamma_n,
\]
and additional (non-Shannon-type) information inequalities are required to describe (or outer bound) the entropic region
\cite{ZhangYeung1997,ZhangYeung1998,DoughertyFreilingZeger2011}.
This strict separation implies the existence of polyhedral (i.e., Shannon-type) polymatroids that are not entropic.

The above separation naturally raises the question studied in this paper.
While not every polyhedral polymatroid is entropic, we show that many of the most frequently used submodular functions in
applications admit explicit entropic constructions.
In particular, the function families we consider---coverage/set cover, facility location, truncations/concave-over-modular,
saturated coverage, and monotone graph-cut objectives---all lie inside $\Gamma_n^\ast$ (up to a fixed scaling).
This provides a direct bridge between classical information theory and the SIM framework
\cite{iyer2021generalized,iyer2021submodular,asnani2021independence}.

\subsection{Examples of known entropic functions.}
Before turning to our constructions, it is useful to recall several classes of submodular functions that are
already known to be entropic or admit immediate entropic interpretations.

First, note that entropic functions are monotone so for a submodular function to be an entropic function, it must be polyhedral (monotone, normalized, and submodular). 

\emph{(Positive) Modular functions} are trivially entropic.
Given weights $(w_i)_{i=1}^n$, the modular function $f(A)=\sum_{i\in A} w_i$ is realized by independent random variables
$X_i$ with $H(X_i)=w_i$, yielding $H(X_A)=\sum_{i\in A} H(X_i)=f(A)$. We need $w_i \geq 0$ for this to hold. 

A closely related example is \emph{rank functions of linear matroids}.
If $U$ is a random vector over a finite field and $X_i$ are linear measurements of $U$, then
$A\mapsto H(X_A)$ equals the rank of the corresponding measurement matrix (up to a $\log q$ factor),
recovering uniform matroid rank functions as entropic.
This construction underlies classical connections between entropy, linear algebra, and matroid theory
\cite{Yeung2008}.

Another canonical example is the \emph{log-determinant} function.
For a zero-mean Gaussian random vector $X\sim \mathcal{N}(0,\Sigma)$, the differential entropy satisfies
\[
h(X_A)=\tfrac12 \log \det(\Sigma_{A,A}) + \text{const},
\]
showing that log-determinant objectives, widely used in experimental design, sensor placement, and Gaussian processes, 
are entropic up to an additive constant.
This provides an important continuous-valued instance of entropic submodularity.

Beyond these cases, relatively few explicit entropic realizations are known for common polyhedral submodular objectives.
In particular, classical combinatorial functions such as coverage, facility location, truncations, and graph-cut–type
objectives are typically treated as abstract submodular functions, without an accompanying Shannon-theoretic
interpretation.
Our results show that many of these widely used objectives also admit exact entropic realizations, despite the fact that
the entropy region is strictly smaller than the polymatroid cone.

\section{Entropic Constructions for Common Submodular Functions}
\label{sec:entropic_constructions}

We use the standard definition of an entropic set function.

\begin{definition}[Entropic set function]
A set function $f:2^{[n]}\to\mathbb{R}_{\ge 0}$ is \emph{entropic} if there exist random variables
$X_1,\dots,X_n$ such that for all $A\subseteq [n]$,
\[
f(A)=H(X_A),\qquad X_A := (X_i)_{i\in A}.
\]
\end{definition}

\noindent\textbf{Connection to SIM.}
If $f(A)=H(X_A)$ is entropic, then the SIM primitives built from $f$ coincide exactly with classical quantities.
For example, $I_f(A;B)=I(X_A;X_B)$ and $H_f(A\mid B)=H(X_A\mid X_B)$, with an analogous identity for conditional mutual
information. This motivates our focus on explicit entropic realizations.

Throughout, we freely use the following elementary fact: if $(Z_t)_{t\in T}$ are mutually independent,
then for any $S\subseteq T$ we have $H((Z_t)_{t\in S})=\sum_{t\in S}H(Z_t)$, and if a random variable
appears multiple times in a tuple it does not increase the joint entropy.

\vspace{0.3em}
\noindent\textbf{Remark on weights.}
For simplicity, we state results with $H(Z)=w$ for arbitrary $w\ge 0$. For discrete Shannon entropy,
one may either (i) scale all weights by a common constant and use uniform random variables on large alphabets
to realize integer entropies exactly, or (ii) state an ``almost-entropic'' version by approximating real
weights arbitrarily well with large alphabets. These standard quantization/scaling steps are omitted below.

In the next sections, we will outline the entropic constructions for the Weighted Coverage/Set Cover, Facility Location function, Truncations and Concave over Modular, Saturated Coverage, and Monotonic Graph Cut function. 

\section{Weighted Coverage / Set Cover is Entropic} \label{sec:set_cover}


\begin{theorem}[Weighted coverage (set cover) is entropic]
\label{thm:coverage_entropic}
Let $U$ be a universe with weights $(w_u)_{u\in U}$ where $w_u\ge 0$. Let the ground set be $\Omega=[n]$
with subsets $U_i\subseteq U$ for $i\in[n]$. Define
\[
f(A) := \sum_{u\in \cup_{i\in A} U_i} w_u,\qquad A\subseteq [n].
\]
Then $f$ is entropic.
\end{theorem}

\begin{proof}
For each $u\in U$, let $Z_u$ be independent random variables with $H(Z_u)=w_u$.
For each $i\in[n]$, define
\[
X_i := (Z_u)_{u\in U_i}.
\]
For any $A\subseteq[n]$, the joint tuple $X_A=(X_i)_{i\in A}$ contains exactly the collection of variables
$\{Z_u: u\in \cup_{i\in A}U_i\}$ (possibly with duplicates). Removing duplicates does not change joint entropy,
and independence yields
\begin{align}
H(X_A) &=H\big((Z_u)_{u\in \cup_{i\in A}U_i}\big) \nonumber \\
&=\sum_{u\in \cup_{i\in A}U_i} H(Z_u) \nonumber \\
&=\sum_{u\in \cup_{i\in A}U_i} w_u \nonumber \\ &=f(A).
\end{align}
\end{proof}

\section{Facility Location is Entropic (Max-by-Nesting Construction)} \label{sec:fac_loc}

We first show that a max-of-weights set function is entropic via nested functional dependence.

\begin{lemma}[Max is entropic via nested suffix variables]
\label{lem:max_entropic}
Let $w_1\ge w_2\ge \cdots \ge w_m\ge 0$ and define $g:2^{[m]}\to\mathbb{R}_{\ge 0}$ by
$g(\emptyset)=0$ and for $A\neq \emptyset$,
\[
g(A):=\max_{a\in A} w_a.
\]
Then $g$ is entropic.
\end{lemma}

\begin{proof}
Let $w_{m+1}:=0$ and define increments $\delta_r := w_r-w_{r+1}\ge 0$ for $r=1,\dots,m$.
Let $B_1,\dots,B_m$ be mutually independent random variables with $H(B_r)=\delta_r$.
Define for each $r\in[m]$,
\[
X_r := (B_r,B_{r+1},\dots,B_m).
\]
Then $H(X_r)=\sum_{t=r}^m H(B_t)=\sum_{t=r}^m \delta_t = w_r$.
For any nonempty $A\subseteq[m]$ let $r^\star=\min(A)$, which is the index with maximum weight since
$(w_r)$ is nonincreasing. The tuple $X_{r^\star}$ contains all variables $(B_t)_{t\ge r^\star}$,
and every $X_a$ for $a\in A$ is a deterministic function of $X_{r^\star}$ (it is a suffix).
Hence $(X_a)_{a\in A}$ contains no additional independent information beyond $X_{r^\star}$, so
\begin{align}
H(X_A) =H(X_{r^\star}) =w_{r^\star} =\max_{a\in A} w_a =g(A).
\end{align}
\end{proof}

\begin{theorem}[Facility location is entropic]
\label{thm:facility_location_entropic}
Let $V=[n]$ and let $S=(s_{ij})_{i,j\in V}$ with $s_{ij}\ge 0$. Define the facility location function
\[
f(A):=\sum_{i\in V} \max_{a\in A} s_{ia},\qquad A\subseteq V,
\]
with the convention $\max_{a\in \emptyset} s_{ia}=0$. Then $f$ is entropic.
\end{theorem}

\begin{proof}
Fix $i\in V$ and consider the multiset of weights $\{s_{ia}:a\in V\}$. Let $\pi_i$ be a permutation of $V$
such that
\[
s_{i,\pi_i(1)}\ge s_{i,\pi_i(2)}\ge \cdots \ge s_{i,\pi_i(n)}\ge 0,
\]
and define $w^{(i)}_r := s_{i,\pi_i(r)}$ for $r\in[n]$, with $w^{(i)}_{n+1}:=0$. Apply Lemma~\ref{lem:max_entropic}
to $(w^{(i)}_r)_{r=1}^n$ to obtain independent variables $(B_{i,r})_{r=1}^n$ with entropies
$H(B_{i,r})=w^{(i)}_r-w^{(i)}_{r+1}$ and the nested ``suffix'' variables
\[
Y_{i,r}:=(B_{i,r},B_{i,r+1},\dots,B_{i,n}),\qquad r\in[n],
\]
satisfying $H((Y_{i,r})_{r\in R})=\max_{r\in R} w^{(i)}_r$ for any $R\subseteq[n]$.

Now define the rank map $r_i(a)\in[n]$ as the unique index such that $a=\pi_i(r_i(a))$.
For each facility $a\in V$, define the random variable
\[
X_a := (Y_{i,r_i(a)})_{i\in V}.
\]
Independence across different $i$ implies that for any $A\subseteq V$,
\[
H(X_A)=\sum_{i\in V} H\big((Y_{i,r_i(a)})_{a\in A}\big).
\]
For fixed $i$, by Lemma~\ref{lem:max_entropic} and the ordering of weights,
\[
H\big((Y_{i,r_i(a)})_{a\in A}\big) = \max_{a\in A} s_{ia}.
\]
Summing over $i$ yields $H(X_A)=\sum_{i\in V}\max_{a\in A}s_{ia}=f(A)$.
\end{proof}

\section{Truncations and Concave-over-Modular} \label{sec:concavemod}

\begin{theorem}[Cardinality truncation $\min(|A|,k)$ is entropic]
\label{thm:card_trunc_entropic}
Fix integers $n\ge 1$ and $1\le k\le n$. Define $t_k(A):=\min(|A|,k)$.
Then $t_k$ is entropic up to a constant factor: there exist random variables
$X_1,\dots,X_n$ such that for all $A\subseteq[n]$,
\[
H(X_A)=\min(|A|,k)\cdot \log q
\]
for some field size $q\ge n$.
\end{theorem}

\begin{proof}
\noindent\textbf{Attribution.}
The proof uses a classical linear-entropy construction: linear measurements of a uniform random vector over a finite
field yield entropic rank functions (linear polymatroids).
Vandermonde matrices are employed to ensure that any subset of size at most $k$ is linearly independent, as in
standard MDS code and secret-sharing constructions.
See \cite{Yeung2008} for a comprehensive treatment of linear entropy functions.

Let $\mathbb{F}_q$ be a finite field with $q\ge n$. Let $U=(U_1,\dots,U_k)$ be i.i.d.\ uniform over
$\mathbb{F}_q$. Choose distinct $\alpha_1,\dots,\alpha_n\in\mathbb{F}_q$ and define for each $i\in[n]$
the linear form
\[
X_i := \sum_{j=0}^{k-1} \alpha_i^{\,j} U_{j+1}\ \in \mathbb{F}_q.
\]
For any $A\subseteq[n]$, the vector $X_A$ is a linear transformation of $U$ with matrix $M_A$ whose rows are
$(1,\alpha_i,\alpha_i^2,\dots,\alpha_i^{k-1})$ for $i\in A$. Any $\ell\le k$ rows of a Vandermonde matrix with
distinct $\alpha_i$ are linearly independent, hence
\begin{equation}
\text{rank}(M_A)=\min(|A|,k).
\end{equation}
Since $U$ is uniform and the transformation is linear, $X_A$ is uniform over an $\mathbb{F}_q$-subspace of
dimension $\text{rank}(M_A)$, implying

\begin{equation}
H(X_A)=\text{rank}(M_A)\cdot \log q = \min(|A|,k)\cdot \log q.
\end{equation}
\end{proof}

\begin{corollary}[Weighted truncation $\min(w(A),k)$ is entropic (integer weights)]
\label{cor:weighted_trunc_entropic}
Let weights $(w_i)_{i=1}^n$ be nonnegative integers and define $w(A)=\sum_{i\in A} w_i$.
Then the truncation $A\mapsto \min(w(A),k)$ is entropic up to a constant factor.
\end{corollary}

\begin{proof}
Create $\sum_{i=1}^n w_i$ ``clones'' partitioned into groups of size $w_i$ for each item $i$.
Apply Theorem~\ref{thm:card_trunc_entropic} on the clone ground set to obtain random variables $(Y_c)$
with $H(Y_S)=\min(|S|,k)\log q$ for all clone subsets $S$.
Define $X_i$ as the tuple of $Y_c$ over clones belonging to $i$.
Then for any $A\subseteq[n]$, $X_A$ corresponds to the clone subset of size $\sum_{i\in A} w_i=w(A)$, hence
\[
H(X_A)=\min(w(A),k)\log q.
\]
\end{proof}

\begin{theorem}[Concave-over-modular from truncation mixtures is entropic]
\label{thm:concave_over_modular_entropic}
Let $w(A)=\sum_{i\in A} w_i$ with $w_i\in\mathbb{Z}_{\ge 0}$. Suppose $g:\mathbb{Z}_{\ge 0}\to\mathbb{R}_{\ge 0}$
is nondecreasing and concave with $g(0)=0$. Then there exist coefficients $(c_t)_{t\ge 1}$ with $c_t\ge 0$ such that
for all integers $x\ge 0$,
\[
g(x)=\sum_{t\ge 1} c_t\,\min(x,t),
\]
and consequently the set function $f(A):=g(w(A))$ is entropic (up to a constant factor).
\end{theorem}

\begin{proof}
Define discrete first differences $\Delta g(t):=g(t)-g(t-1)$ for $t\ge 1$.
Concavity and monotonicity imply $\Delta g(t)\ge 0$ and $\Delta g(t)$ is nonincreasing in $t$.
Let $c_t := \Delta g(t)-\Delta g(t+1)\ge 0$ (with $\Delta g(t)\to 0$ for large $t$, or define $\Delta g(T+1)=0$
beyond the maximum relevant $T$).
Then for any $x\ge 0$,
\[
\sum_{t\ge 1} c_t\,\min(x,t)
=\sum_{t=1}^{x} c_t\,t + \sum_{t>x} c_t\,x
=\sum_{r=1}^{x} \Delta g(r)
=g(x),
\]
where the last equality follows by telescoping $g(x)=\sum_{r=1}^x (g(r)-g(r-1))=\sum_{r=1}^x \Delta g(r)$.
Thus $g$ is a conic combination of truncations.
By Corollary~\ref{cor:weighted_trunc_entropic}, each $A\mapsto \min(w(A),t)$ is entropic (up to scaling).
Using independent blocks of random variables for each $t$ and concatenating them shows that a nonnegative sum of
entropic functions is entropic (after matching common scaling), hence $f(A)=g(w(A))$ is entropic.
\end{proof}

\section{Saturated Coverage} \label{sec:saturatedcov}

\begin{theorem}[Saturated coverage is entropic]
\label{thm:saturated_coverage_entropic}
Let $U$ be a universe with weights $(w_u)_{u\in U}$ and sets $U_i\subseteq U$ for $i\in[n]$.
Let $\kappa\ge 0$. Define the \emph{saturated coverage} function
\[
f_{\mathrm{sat}}(A) := \min\Big(\sum_{u\in \cup_{i\in A} U_i} w_u,\ \kappa\Big).
\]
Then $f_{\mathrm{sat}}$ is entropic (up to a constant factor), assuming $\kappa$ is compatible with the
weight scaling (e.g., integer after a common scaling).
\end{theorem}

\begin{proof}
By Theorem~\ref{thm:coverage_entropic}, there exist random variables $(X_i)_{i=1}^n$ such that
$H(X_A)=\sum_{u\in \cup_{i\in A}U_i} w_u$.
Let $W(A):=H(X_A)$ denote this (modular-in-components) entropy.
By Corollary~\ref{cor:weighted_trunc_entropic} with modular function $W(A)$ (after a common scaling to integers),
the truncation $A\mapsto \min(W(A),\kappa)$ is entropic up to scaling.
Since $\min(W(A),\kappa)=f_{\mathrm{sat}}(A)$ by definition, the claim follows.
\end{proof}

\section{Monotone Graph-Cut--Like Function is Entropic} \label{sec:graph-cut}

\begin{theorem}[Monotone graph-cut--like objective is entropic for $\lambda\le \tfrac12$]
\label{thm:monotone_graphcut_entropic}
Let $V=[n]$ and let $S=(s_{ij})$ be symmetric with $s_{ij}\ge 0$ and $s_{ii}=0$.
For $\lambda\in[0,\tfrac12]$, define
\[
F(A):=\sum_{i\in A}\sum_{j\in V} s_{ij} \;-\; \lambda \sum_{i\in A}\sum_{j\in A} s_{ij},\qquad A\subseteq V.
\]
Then $F$ is entropic.
\end{theorem}

\begin{proof}
View $s_{ij}$ as weights on undirected edges $e=\{i,j\}$ for $i<j$, and write $w_e:=s_{ij}$.
For each edge $e=\{i,j\}$, create three mutually independent random variables:
a shared variable $Z_e$ and two endpoint-private variables $U_{e,i}$ and $U_{e,j}$, with entropies
\[
H(Z_e)=2\lambda\,w_e,\qquad H(U_{e,i})=H(U_{e,j})=(1-2\lambda)\,w_e.
\]
(All are nonnegative since $\lambda\in[0,\tfrac12]$.) Assume all such variables are independent across edges.

Define for each vertex $i\in V$,
\[
X_i := \Big((Z_e)_{e\ni i},\ (U_{e,i})_{e\ni i}\Big),
\]
i.e., $X_i$ contains the shared variable for every incident edge and the private variable corresponding
to that vertex on each incident edge.

Fix $A\subseteq V$. Since all base variables are independent across edges, $H(X_A)$ equals the sum over edges
of entropies of the distinct base variables included by selecting endpoints in $A$. For a fixed edge $e=\{i,j\}$:

\begin{itemize}
\item If $|\{i,j\}\cap A|=0$, $X_A$ includes none of $\{Z_e,U_{e,i},U_{e,j}\}$, contributing $0$.
\item If $|\{i,j\}\cap A|=1$, say $i\in A$, then $X_A$ includes $\{Z_e,U_{e,i}\}$, contributing
$H(Z_e)+H(U_{e,i})=2\lambda w_e+(1-2\lambda)w_e=w_e$.
\item If $|\{i,j\}\cap A|=2$, then $X_A$ includes $\{Z_e,U_{e,i},U_{e,j}\}$, contributing
$H(Z_e)+H(U_{e,i})+H(U_{e,j})=2\lambda w_e+2(1-2\lambda)w_e=2(1-\lambda)w_e$.
\end{itemize}

On the other hand, expanding $F(A)$ edge-by-edge for symmetric $s_{ij}$ shows that an edge contributes:
$w_e$ if exactly one endpoint is in $A$, and $2(1-\lambda)w_e$ if both endpoints are in $A$.
Thus the per-edge contributions match, and summing over edges yields $H(X_A)=F(A)$ for all $A$.
\end{proof}


\section{Conclusion}
This paper addresses a basic question at the interface of information theory and submodular optimization: \emph{when do the submodular functions that dominate modern applications admit genuine Shannon-theoretic interpretations?}
While Shannon entropy is a canonical source of submodularity \cite{Shannon1948,Yeung2008}, the entropy region is strictly smaller than the polymatroid cone for $n\ge 4$ \cite{ZhangYeung1997,ZhangYeung1998,DoughertyFreilingZeger2011,Matus2013AlmostEntropic}, implying the existence of polyhedral submodular functions that are not entropic.
Despite this separation, our results show that a broad family of submodular objectives used most frequently in practice \emph{are} entropic and admit simple, explicit probabilistic constructions.

We provided constructive entropic realizations for weighted coverage/set cover, facility location, truncations and concave-over-modular functions, saturated coverage, and a monotone graph-cut–type objective.
Each construction is modular and interpretable: coverage arises from unions of independent element variables, facility location from nested functional dependence realizing max operators, truncations from linear-algebraic rank constructions and concave mixtures, and monotone graph cut from a shared/private edge decomposition.

A central implication is conceptual: for these function classes, \emph{combinatorial information measures reduce exactly to classical information measures}.
In SIM frameworks, submodular mutual information, conditional gain, and submodular conditional mutual information are defined by replacing entropy with a generic submodular function \cite{iyer2021generalized,iyer2021submodular,asnani2021independence}.
Our entropic realizations show that, for coverage, facility-location, truncation, and graph-based objectives, these quantities coincide with Shannon-theoretic mutual information, conditional entropy, and conditional mutual information.
Thus, the information-theoretic language used in SIM-based applications is literal rather than metaphorical, providing a principled bridge between the classical geometry of entropy and polymatroids \cite{Yeung2008,ZhangYeung1998} and submodular objectives used in data-centric learning pipelines \cite{kothawade2021similar,kothawade2022prism,karanam2022orient,li2022platinum,majee2024score,beck2024theoretical}.

More broadly, our results help explain why a small set of classical submodular functions has repeatedly emerged as effective modeling primitives: these functions lie naturally inside the entropic region and inherit information-theoretic identities and intuition.
Beyond conceptual clarity, explicit entropic constructions may enable new analysis tools for submodular objectives by importing classical inequalities, chain rules, and decomposition arguments into SIM-based learning settings.

Several directions remain open.
Mapping a more complete landscape of practical submodular functions into (or outside) the entropic region, sharpening parameter thresholds in graph-based constructions, and characterizing which compositions preserve entropic realizability are all promising avenues.
Overall, although the entropy region is strictly smaller than the polymatroid cone in general, our results show that many of the submodular objectives that matter most in practice lie squarely within it, strengthening the foundations of combinatorial information measures and their connection to classical information theory.

\bibliography{refs}
\bibliographystyle{abbrv}
\end{document}